\documentclass[letterpaper, 10pt, conference]{ieeeconf}  % Comment this line out if you need a4paper

\makeatletter
\def\endthebibliography{%
  \def\@noitemerr{\@latex@warning{Empty `thebibliography' environment}}%
  \endlist
}
\makeatother

\IEEEoverridecommandlockouts                              % This command is only needed if 
\usepackage[section]{placeins}
\usepackage{algorithmicx}
\usepackage{algorithm}
\usepackage{graphicx}
\usepackage{algpseudocode}
\usepackage{amsmath}
\usepackage{caption}
\usepackage{multicol}
\usepackage{amssymb}

\newtheorem{definition}{Definition}
\newtheorem{remark}{Remark}
\newtheorem{proposition}{Proposition}
\newtheorem{assumption}{Assumption}
\usepackage{subfigure}
\newcommand{\real}{\mathbb{R}}

\title{\LARGE \bf
Self-triggered Control for Safety Critical Systems using Control Barrier Functions
}

\author{Guang Yang$^{1}$,  Calin Belta$^{2}$ and Roberto Tron$^{2}$% <-this % stops a space
\thanks{This material is based upon work partially supported by the National Science Foundation under Grant NSF IIS-1723995 and NSF CMMI-1728277. }% <-this % stops a space
\thanks{$^{1}$Guang Yang is with Division of Systems Engineering at Boston University, Boston, MA 02215 USA. Email: {\tt\small gyang101@bu.edu}}
\thanks{$^{2}$Roberto Tron,$^{2}$Calin Belta are with the Department of Mechanical Engineering at Boston University, Boston, MA 02215 USA. Email: {\tt\small tron@bu.edu, cbelta@bu.edu}}
}

\begin{document}

\maketitle

%%%%%%%%%%%%%%%%%%%%%%%%%%%%%%%%%%%%%%%%%%%%%%%%%%%%%%%%%%%%%%%%%%%%%%%%%%%%%%%%
\begin{abstract}
We propose a real-time control strategy that combines self-triggered control with Control Lyapunov Functions (CLF) and Control Barrier Functions (CBF). Similar to related works proposing CLF-CBF-based controllers, the computation of the controller is achieved by solving a Quadratic Program (QP). However, we propose a Zeroth-Order Hold (ZOH) implementation of the controller that
overcomes the main limitations of traditional approaches based on periodic controllers, i.e., unnecessary controller updates and potential violations of the safety constraints. Central to our approach is the novel notion of safe period, which enforces a strong safety guarantee for implementing ZOH control. In addition, we prove that the system does not exhibit a Zeno behavior as it approaches the desired equilibrium. 
\end{abstract}

%%%%%%%%%%%%%%%%%%%%%%%%%%%%%%%%%%%%%%%%%%%%%%%%%%%%%%%%%%%%%%%%%%%%%%%%%%%%%%%%

\section{Introduction}

Real-time control is central to many cyber-physical systems, such as autonomous cars, building automation systems, and robots. The design of a real-time controller requires the consideration of several factors, including computational resource constraints, actuator limitations, and safety. In this work, we focus on two complementary goals: stability and safety. 

An effective way to address these two objectives is to use Control Lyapunov Functions (CLF) \cite{khalil1996noninear} for stability and Control Barrier Functions (CBF) \cite{wieland2007constructive} for safety. This formalism was first used in adaptive cruise control \cite{ames2014control}. It was also adopted in other safety-critical applications, such as lane keeping in autonomous driving \cite{ames2017control}, quadrotor control\cite{wang2017safe} and control of bipedal robot walking\cite{hsu2015control}. The recently introduced notion of Exponential Control Barrier Function (ECBF) \cite{nguyen2016exponential} greatly reduced the complexity of designing CBFs for systems with higher relative degree, as compared to \cite{hsu2015control}, \cite{wu2016safety}. These works compute the desired control using simple optimization problems (typically Quadratic Programs), where the stability and safety requirements are encoded as linear constraints, even for non-linear systems. This formalism, however, is based on a continuous time formulation, which is in contradiction with the reality that these controllers are implemented on digital platforms, where the updates to the control law can happen only at discrete times. In this paper, we address the problem of implementing a continuous CLF-CBF controller on a digital platform with discrete time updates, while still satisfying the stability and safety properties. 

Traditionally, digital controllers are implemented  using 
discretized periodic control inputs. 
A popular discretization method is the Zeroth-Order Hold (ZOH), where the controller computes and holds a discrete control signal for a fixed time period. 
This approach has two potential major drawbacks when naively combined with the CLF-CBF formalism for safety-critical applications. First, given a fixed update period, there is no guarantee that the safety constraints will hold. Since the plant is sampled at a fixed frequency, the system could violate the safety constraints in between two sampled time instances. Second, there are unnecessary computations and control updates due to fixed-time sampling. This is cumbersome for a system with constrained computational resources and actuator life. 

In this paper, we propose to use self-triggered control \cite{anta2010sample} to address these issues. Self-triggered control was introduced in \cite{velasco2003self}, and related works include \cite{mazo2009input}, \cite{mazo2010iss},\cite{wang2009self} and \cite{anta2010sample}. The core of all self-triggered controllers consists of two parts. First, a designed feedback controller computes the control input at a given time instance. Second, it determines the next controller update time instance based on current sensor measurements and mission requirements. This approach, also known as proactive control, is different from event-based control, where the controller is updated when an event occurs. 

In this work, we propose a novel self-triggered controller that pre-computes the next update time instance given the current state, control objective, and safety requirements. While the controller is applied in a ZOH manner, we ensure that the system will not violate the constraints between updates. 

The remaining of the paper is organized as follows. In Section \ref{Preliminaries}, we introduce Control Barrier Functions (CBFs), Control Lyapunov Functions (CLFs), Zeroth-Order Hold (ZOH) control, the system dynamics used throughout the paper. We formulate the problem in Section \ref{ProblemFormulation}. The technical details of the solution are presented in Section \ref{TechnicalApproach}. In Section \ref{ExampleStudy}, we validate our controller on a double integrator dynamical system. In particular, we empirically compare our self-triggered control strategy with standard periodic control. We conclude and discuss future directions in Section \ref{Conclusion}.

\section{Preliminaries} \label{Preliminaries}
\subsection{Notation}
We use $\mathbb{Z}$
and $\mathbb{R}^n$ to denote the set of integers and the set of real numbers in $n$ dimensions, respectively. The Lie derivative of a smooth function $h(x(t))$ along dynamics $\dot{x}(t)=f(x(t))$ is denoted as $\pounds_{f} h(x) := \frac{\partial h(x(t))}{\partial x(t)} f(x(t))$. We use $\pounds_{f}^{r_b} h(x)$ to denote a Lie derivative of higher order $r_b$, where $r_b \geq 0$. A function $f: \mathbb{R}^n \mapsto \mathbb{R}^m$ is called \textit{Lipschitz continuous} on $\mathbb{R}^n$ if there exists a positive real constant $L \in \mathbb{R}^+$, such that $\|f(y)-f(x)\| \leq L \|y-x\|, \forall x,y \in \mathbb{R}^n$. Given a smooth function $h:\mathbb{R}^n \mapsto \mathbb{R}$, we denote $h^{r_b}$ as its $r_b$-th derivative with respect to time $t$. A continuous function $\alpha:[-b,a) \mapsto [-\infty,\infty)$, for some $a>0, b>0$, belong to extended class $K$ if $\alpha$ is strictly increasing on $\mathbb{R}^+$ and $\alpha(0)=0$. 

\subsection{Safety Constraints and Control Barrier Functions}
Consider a continuous time dynamical control system
\begin{equation}\label{eq:dynamicSystem}
\dot{x} = f(x) + g(x)u,
\end{equation} where $x \in \mathbb{R}^n$, $u \in \mathbb{R}^m$, and $f(x)$, $g(x)$ are locally Lipschitz continuous. Let $x_0:=x(t_0) \in \mathbb{R}^n$ denote the initial state. For any initial condition $x_0$, there exists a maximum time interval $\mathrm{I}(x_0) = [t_0,t_{max})$ such that $x(t), \forall t \in \mathrm{I}(x_0)$ is a unique solution. 
Next, we define a set of safety constraints. Given a continuously differentiable function $h: \mathbb{R}^n \mapsto \mathbb{R}$, we define a closed safety set $C$:
\begin{equation}\label{eq:safetySet}
\begin{aligned}
C &= \{x \in{\mathbb{R}^n}|h(x)  \geq 0 \}. \\
\partial{C} &= \{x \in \mathbb{R}^n|h(x)  = 0 \}, \\
Int(C) &= \{x \in \mathbb{R}^n|h(x)  > 0 \}.
\end{aligned}
\end{equation}
The set $C$ is called \emph{forward invariant} for system (\ref{eq:dynamicSystem})
if $x_0 \in C$ implies $x(t)\in C, \forall t \in \mathrm{I}(x_0)$. 

Given a continuously differentiable $h:\mathbb{R}^n \mapsto \mathbb{R}$ and dynamics \eqref{eq:dynamicSystem}, the relative degree $r_b \geq 0$ is defined as the smallest natural number such that $\pounds_{g} \pounds_{f}^{r_b-1} h(x) u \neq 0$. The time derivative of $h$ are related to the Lie derivatives by:

\begin{equation} \label{eq:rbh}
h^{r_b}(x) = \pounds_{f}^{r_b}h(x) + \pounds_{g}\pounds_{f}^{r_b-1}h(x)u.
\end{equation}

To ensure forward invariance for systems with higher relative degrees, \cite{nguyen2016exponential} introduced the notion of Exponential Control Barrier Function (ECBF). Before formally reviewing its definition, a transverse variable is defined as
\begin{align}
\xi_b(x) = \left[ \begin{matrix} h(x)\\ 
\dot{h}(x)\\
\cdot \\
\cdot \\
h^{r_b}(x) \end{matrix} \right],
\end{align}
together with a virtual control
\begin{equation}
\mu = (\pounds_{g}\pounds_{f}^{r_b-1}h(x))^{-1}(\mu-\pounds_{f}^{r_b}h(x)).
\end{equation}
The input-output linearized system corresponding to \eqref{eq:dynamicSystem} is
\begin{align*}
\dot{\xi}_b(x) = A_b \xi_b(x) + B_b \mu, \\
y = C_b \xi_b(x)= h(x),
\end{align*}
with 
\begin{align}
A_b = \left[\begin{matrix}0 &1 &\cdot &\cdot &0 \\
0 &0 &1 &\cdot &0\\
\cdot &\cdot &\cdot &\cdot &\cdot\\
0 &0 &0 &\cdot &1\\
0 &0 &0 &0 &0 \end{matrix}\right], B_b = \left[\begin{matrix}0\\
\cdot\\
\cdot\\
0\\
1\end{matrix}\right],\\
C_b = \left[1 \cdot \cdot \cdot 0\right].
\end{align}

\begin{definition}
(\textit{Zeroing Control Barrier Function}) Consider a dynamical system in \eqref{eq:dynamicSystem} and the closed set $C$ defined in \eqref{eq:safetySet}. Given a continuously differentiable function $h:\mathbb{R}^n \mapsto \mathbb{R}$ with relative degree $r_b = 1$. If there exits a locally Lipschitz extended class $\mathrm{K}$ function $\alpha$ and a set $C$, such that 
\begin{align}\label{eq:zcbfConstraint}
\inf \limits_{u\in U} [\pounds_{f} h(x)+&\pounds_{g}h(x)u+\alpha(h(x))] \geq 0, \quad \forall x \in Int(C),
\end{align}
then $h(x)$ is a zeroing control barrier function (ZCBF)\cite{xu2016robustness} and it implies \textit{forward invariance} of system \eqref{eq:dynamicSystem}.
\end{definition}

\begin{definition}
(\textit{Exponential Control Barrier Function}) 
Consider a dynamical system \eqref{eq:dynamicSystem}, the safety set $C$ defined in \eqref{eq:safetySet} and $h(x)$ with relative degree $r_b \geq 1$. Then $h(x)$ is an exponential control barrier function (ECBF)\cite{nguyen2016exponential} if there exists $K_b \in \mathbb{R}^{1\times r_b}$, such that
\begin{equation} \label{eq:ecbfConstraint}
\inf \limits_{u\in U} [\pounds_{f}^{r_b} h(x)+\pounds_{g}\pounds_{f}^{r_b-1} h(x)u+ K_b \xi_b(x) ] \geq 0, \forall x \in Int(C).
\end{equation}
The row vector of coefficients $K_b$ is selected such that the closed-loop matrix $A_b -B_b K_b$ has all negative real eigenvalues.

\begin{remark}\label{remark1}
As pointed out in \cite{nguyen2016exponential}, the ZCBF is a special case of ECBF with relative degree $r_b=1$.
\end{remark}

\end{definition}

\subsection{Stabilization with Control Lyapunov Function}
\begin{definition} 
(\textit{Exponentially-Stabilizing Control Lyapunov Function}) Given the system \eqref{eq:dynamicSystem}, a continuously differentiable function $V: \mathbb{R}^n \mapsto \mathbb{R}$ is an Exponentially-Stabilizing Control Lyapunov Function (ES-CLF)\cite{ames2014rapidly} if there exists positive constants $c_1,c_2,\epsilon \geq 0$, such that

\begin{equation}
\begin{aligned}\label{eq:CLF}
\centering
&c_1 \|x\|^2\leq V(x)  \leq c_2 \|x\|^2, \\
& \inf \limits_{u\in U} [\pounds_{f} V(x)+\pounds_{g}V(x)u+\epsilon V(x)] \leq 0, \quad \forall x \in \mathbb{R}^n.
\end{aligned}
\end{equation}
\end{definition}
The existence of a ES-CLF implies that there exists a set of controllers 
\begin{equation*}
K_{ES-CLF} = \{u\in U: \pounds_{f} V(x)+\pounds_{g}V(x)u+\epsilon V(x)] \leq 0\},
\end{equation*}
such that the system is exponentially stabilized \cite{ames2014rapidly}, i.e.
\begin{equation}
x(t) \leq \sqrt[]{\frac{c_2}{c_1}}e^{-\frac{\epsilon}{2}t}\|x_0\|, \quad \forall t \geq 0
\end{equation}

\subsection{Zeroth-Order Hold}
The Zeroth-Order Hold (ZOH) control  mechanism holds the control signal at $t_k$ over a period of time, i.e. $u(s) = u(t_k), \forall s \in [t_k,t_{k+1})$. The sequence of control update time instants $\{t_k\}_{k \in \mathbb{N}}$ is strictly increasing. 

\section{Problem formulation}\label{ProblemFormulation}
Let the continuous dynamical system defined in \eqref{eq:dynamicSystem} with an initial state $x_0 \in Int(C)$. The goal is to stabilize the system to a desired state $x_{d} \in \mathbb{R}^n$ under discretized control input while guaranteeing \textit{forward invariance} of the safety set defined in \eqref{eq:safetySet}.
We propose self-triggered controller that uses a Quadratic Program (QP) to compute the control signal, and that actively computes the next update instance given the safety constraints and control objective. In particular, we introduce the notions of a safe periods for the safety constraints ($\tau_{\mathrm{CBF}}$) and for the stability constraints ($\tau_{\mathrm{CLF}}$). These safe periods are computed by means of a lower bound on the ECBF constraints, upper bounds on the CLF, and bounds on the trajectories of the system (i.e., we do not require an explicit integration of the dynamics \eqref{eq:dynamicSystem}.

%\begin{remark} \label{remakr2}
%Based on \eqref{eq:dynamicSystem}, it might be difficult to find a closed-form solution of $x(t), \forall t\in \mathrm{I(x_0)}$. Therefore, we will employ a bound on the system trajectories that is only time dependent to evaluate the safety properties of the system. 
%\end{remark}

\section{Self-triggered Control using CBF}\label{TechnicalApproach}
In this section, we define the CLF-CBF QP for our controller. Next, the notion of safe periods for the CBF and CLF constraints is introduced. Lastly, we present the complete controller update strategy. 
\subsection{CBF-CLF Quadratic Program formulation}
Given \eqref{eq:dynamicSystem},  the CBF-CLF QP is defined as
\begin{equation}
\begin{aligned}
& \underset{u\in \mathrm{U}}{\text{min}}
& & u^Tu \\
& \text{s.t.}
& & \pounds_{f}^{r_b} h(x)+\pounds_{g}\pounds_{f}^{r_b-1} h(x)u+ K_b \xi_b \geq 0,\\
& & &\pounds_{f} V(x)+\pounds_{g}V(x)u+\epsilon V(x) \leq 0,\\
& & &x(t_k)\in Int(C),\\
%& & & u_{l} \geq u \geq u_{u}.
\end{aligned}
\label{eq:ECBF-ZCBF-CLF}
\end{equation}
The control input is constrained to be in a convex set $U$, which can be used to model practical actuation limits (e.g., for $u\in\real{}$, we might have lower and upper bounds $u_{l}$ and  $u_{u}$, respectively).

At every update instance $t_k$,  we solve \eqref{eq:ECBF-ZCBF-CLF} to compute the optimal control input $u_k$. This control is applied in a ZOH manner until the next update instance $t_{k+1}$. 
At a high level, the strategy used by our self-triggered controller is to evaluate whether, with $u_k$ applied in a ZOH manner, the ECBF constraint in \eqref{eq:ECBF-ZCBF-CLF} will still hold in the interval $t_{k+1} \geq t\geq t_k$, and whether the CLF will decrease after at the end of the same period.

\subsection{Distance bound on a system trajectory}
For the computation of the safe periods for the ECBF constraints, we rely on bounds for the inequalities in \eqref{eq:ECBF-ZCBF-CLF}. Since these inequalities are state-dependent, we need a simple way to describe the trajectory of the system \eqref{eq:dynamicSystem} (since, in general, an exact integration of the dynamics might be computationally infeasible for a real-time controller). More specifically, we propose to find a bound on the system trajectory that exclusively depends on general properties of the system dynamics. Since we evaluate the trajectory bound at every $t_k$, we denote $r_{t_k}(t) = r(t+t_k), \forall t \geq t_k$. The upper bound of $r_{t_k}$ is defined as $\overline{r}_{t_k}$.

\begin{proposition}\label{proposition}
Given the dynamical system defined in \eqref{eq:dynamicSystem}, starting at $x(t_k)$ the distance between the trajectory $x(t+t_k)$ and $x(t_k)$ is bounded by $\overline{r}_{t_k}(t) = r_{0}e^{L(t-t_k)} - \frac{1}{L}\|f(x(t_k))+g(x(t_k))u_k\|, \forall t \geq t_k$. 
\end{proposition}
\begin{proof}
Let $r_{t_k}(t) = \|x(t_k+t) - x(t_k)\|$. Its derivative $\dot{r}(t)$ can be calculated as
\begin{eqnarray*}
\dot{r}(x(t+t_k)) &=& \frac{(x(t+t_k)-x(t_k))^T}{\|x(t+t_k)-x(t_k)\|} \dot{x}(t+t_k) \\
&=&  \frac{(x(t+t_k)-x(t_k))^T}{\|x(t+t_k)-x(t_k)\|} {f}(x(t+t_k),u). 
\end{eqnarray*}
Since $\frac{(x(t+t_k)-x(t_k))}{\|x(t+t_k)-x(t_k)\|}$ is a unit vector, we have
\begin{equation}
\begin{aligned}\label{eq:rboundNobar}
&\dot{r}_{t_k} \leq \|{f}(x(t+t_k),u)\| \\
&\leq \| {f}(x(t+t_k),u) - {f}(x(t_k),u) + {f}(x(t_k),u)\| \\
&\leq \| f(x(t+t_k),u) - f(x(t_k),u)\|+ \|f(x(t_k),u)\|,
\end{aligned}
\end{equation}
where we used the triangular inequality. Because of the assumption on the Lipschitz continuity of the system dynamics, the following condition holds $\|  f(x(t+t_k),u) - f(x(t_k),u)\| \leq L \|x(t+t_k) - x(t_k)\|$, with $L$ as the Lipschitz constant for $f$. By plugging the inequality into \eqref{eq:rboundNobar}, we get 
\begin{align}\nonumber
\overline{\dot{r}}_{t_k}(t) &= L\|x(t+t_k) - x(t_k)\| + \|  f(x(t_k),u)\| \\ \label{eq:rBound}
&\leq L\overline{r}(t+t_k) + \|  f(x(t_k),u)\|.
\end{align}
In this case, $ \| f(x(t_k),u)\| = \|f(x(t_k))+g(x(t_k))u_k\|$. The solution for \eqref{eq:rBound} is 
\begin{equation}\label{eq:trajBound}
\overline{r}_{t_k}(t) = r_{0}e^{L(t-t_k)} - \frac{1}{L}\|f(x(t_k))+g(x(t_k))u_k\|.
\end{equation}
The constant $r_{0}$ is determined by the condition $\overline{r}_{t_k}(0)=r_{t_k}(0)$, that is
\begin{equation*}
r_0 =  \frac{1}{L}\|f(x(t_k))+g(x(t_k))u_k\|.
\end{equation*}
We then have $r_{t_k}< \overline{r}(t_k)$, thanks to the comparison theorem. 
\end{proof}
Once we have $\overline{r}_{t_k}(t)$, we can define a ball that bounds the trajectory under system dynamics \eqref{eq:dynamicSystem} as
\begin{equation*}
B_{\overline{r}_{t_k}} = \{x\in \mathbb{R}^n:\| x(t) - x(t_k)\| \leq \overline{r}_{t_k}\}.
\end{equation*}

\subsection{CBF Safe Period}
\begin{definition}
(\textit{Safe Period}) 
For a dynamical system in \eqref{eq:dynamicSystem}, starting at $x(t_k) \in Int(C)$, if there exists a $\tau_{\mathrm{CBF}}$ such that $x(t_k+\tau_{\mathrm{CBF}}) \in Int(C)$ under a constant control input $u_k$, then $[t_k,t_k+\tau_{\mathrm{CBF}}]$ is the safe time window for the system at $t_k$, and $\tau_{\mathrm{CBF}}$ is the safe period of this system.
\end{definition}

Based on \eqref{eq:ecbfConstraint}, we define the ECBF constraint as
\begin{equation}\label{eq:zetaEFunc}
\zeta_{\mathrm{ECBF}}(x(t)) = \pounds_{f}^{r_b}h(x)+\pounds_{g}\pounds_{f}^{r_b-1} h(x)u+ K_b \xi_b(x)
\end{equation}
for $x\in Int(C).$
Based on \eqref{eq:ecbfConstraint}, the system is \textit{forward invariant} if and only if $\mathrm{\zeta_{\mathrm{ECBF}}}(x(t)) \geq 0 $. We can determine Safe Period $\tau_{\mathrm{CBF}}$ by evaluating the inequality above. 

Next, we use $\overline{r}_{t_k}(t)$ to obtain lower bound $\underline{\zeta}_{\mathrm{ECBF}}(x(t))$ so we do not rely on the closed-form solution of $x(t)$ to evaluate the safety of the system \eqref{eq:dynamicSystem}; In other words, we rely on the implication
\begin{equation*}
\underline{\zeta}_{\mathrm{ECBF}}(t) \geq 0 \implies \zeta_{\mathrm{ECBF}}(x(t)) \geq 0, \forall t_{k+1}\geq t \geq t_k,
\end{equation*}
At an update instance $t_k$, we define the initial condition $\underline{\zeta}_{\mathrm{ECBF}}(t_k) = \zeta_{\mathrm{ECBF}}(x(t_k))$. For a shorter notation, we define $\zeta(x(t)) := \zeta_{\mathrm{ECBF}}(x(t))$. Then $\underline{\zeta}(t)$ can be obtained again by using the comparison theorem with the following:
\begin{equation} \label{eq:zetaBound}
\underline{\zeta}(t) = \underline{\dot{\zeta}}(t)t+\zeta(t_k), \quad
\end{equation}
where $\underline{\dot{\zeta}}(t) \leq \dot{\zeta}(t), \forall t_{k+1} \geq t \geq t_k$. To find $\underline{\dot{\zeta}}(t)$, we first denote the derivative of $\zeta(x(t))$ as
\begin{equation*}
    \dot{\zeta}(x(t)) = \frac{\partial \zeta(x(t))}{\partial x}(f(x(t))+g(x(t))u)
\end{equation*}
After factoring out each term i.e., $\frac{\partial \zeta(x(t))}{\partial x}f(x(t))$ and  $\frac{\partial \zeta(x(t))}{\partial x}g(x(t))u$, we will get an expression in terms of state $x(t)$ and control $u$. Since control $u$ is constant under ZOH, we only need to consider the bound on the state. By using proposition \eqref{proposition}, we can use $\overline{r}_{t_k}(t)$ to bound the state and get $\underline{\zeta}(t)$.
\begin{remark}
Notice $\underline{\zeta}(t)$ is time dependent because we replace state $x(t)$ with $r(t)$ in our original safety constraint $\zeta(x(t))$. We do not need to calculate a closed-form solution from \eqref{eq:dynamicSystem} to evaluate the safety constraint.
\end{remark}

With lower bound $\underline{\zeta}(t)$, we can determine safe period $\tau_{\mathrm{CBF}}$, such that $\underline{\zeta} (t_k+\tau_{\mathrm{CBF}})=0$. The problem is equivalent to finding a root for $\underline{\zeta}$.If the closed-form solution of \eqref{eq:zetaBound} in terms of $t$ is difficult to obtain, we can use algorithms such as secant method \cite{brent2013algorithms} to find its roots. If there are multiple CBF constraints, we denote $i$-th constraint to be $\zeta_i$. The safe period that satisfies all CBF constraints is 
\begin{equation} \label{eq:globalSafePeriod}
\tau_{\mathrm{CBF}} = min(\tau_{\mathrm{CBF},i}), \forall i.
\end{equation}

\subsection{CLF Update Period}
In addition to the safety constraints \eqref{eq:zetaEFunc}, the CLF constraint also needs to be considered for determining the next update time $t_{k+1}$. There could be cases where the system violates the stability constraint while using the ZOH control $u_k$ for $\tau_{CBF}$. Intuitively, the resulting trajectory might overshoot the equilibrium if we naively apply the following update rule $t_{k+1}:=t_k+\tau_{CBF}$. 

Because the QP formulation is solved point-wise in time, we cannot guarantee the property of exponential convergence to the desired state on a ZOH implementation. To achieve at least asymptotic stability, we need to define a CLF update period which guarantees that the Lyapunov function decreases at every step. 

\begin{definition}\label{def:CLFupdatePeriod}
(\textit{CLF Update Period}) 
For the dynamical system defined in \eqref{eq:dynamicSystem}, the $\tau_{\mathrm{clf}}$ is a CLF update period, if $ V(x(t_k+\tau_{\mathrm{CLF}})) - V(x(t_k)) \leq 0$. 
\end{definition}

For systems that do not have a closed-form solution for their trajectories, we need to find a upper bound $\overline{V}(t)$ such that $\overline{V}(t) \geq V(x(t)), \forall t_{k+1} \geq t\geq t_k$,

\begin{equation}
\overline{V}(t) \leq 0 \implies V(x(t)) \leq 0, \forall t_{k+1} \geq t \geq t_k.
\end{equation}
% Therefore, $\tau_{\mathrm{CLF}}$ is a CLF update period if $\overline{\eta}(t_k+\tau_{\mathrm{CLF}})=0, \forall t_k \in \mathrm{I(x_0)}$. 

The upper bound for $V(x(t))$ can be found using descent lemma \cite{bertsekas1999nonlinear}. The following inequality holds $\forall t_{k+1} \geq t\geq t_k$
\begin{align}\label{eq:vBar}
V(x(t)) \leq V(t_k)+(t-t_k)V'(t_k)+(t-t_k)^2\frac{D}{2}=\overline{V}(t).
\end{align}
where $D := \max_{x \in Int(C)} V''$. The proof is in \cite{bertsekas1999nonlinear}. 

\begin{remark}
We can get sharper bounds on D by maximizing the second derivative on $C \cap\{x:V(x)<V(x(t_k))\}$.
\end{remark}

\begin{remark}
We use different bounds for computing $\tau_{CBF}$ and those used for $\tau_{CLF}$ because the constraints would likely start with a zero margin i.e.,$\eta(x(t_k))$ near the equilibrium. This implies $\tau_{CLF}=0$.
\end{remark}

Since $\overline{V}(t)$ is a quadratic function in terms of $t$, there exits a closed-form solution for the roots. Given the condition that we want to enforce when determining $\tau_{CLF}$ shown as below
\begin{equation*}
    \overline{V}(t) - V(x(t_k)) \leq 0.
\end{equation*}
The non-zero root can be expressed as
\begin{equation}\label{eq:tau_update_clf}
    \tau_{CLF} = \frac{-2V'(x(t_k))}{D}.
\end{equation}

\begin{assumption}\label{assumption1}
By using the following inequality constraint $V'(x(t)) \leq -\epsilon V(x(t))$ defined in \eqref{eq:ECBF-ZCBF-CLF}, we assume there is a neighborhood of equilibrium such that for the optimal solution from solving the QP, the inequality shown above becomes equality. Note we expect that this assumption is valid given the nature of QP i.e., satisfying the CLF constraint while minimizing control effort. As the system approaches to equilibrium, the Lyapunov Function $V(x)$ decreases toward zero and optimal control input $u$ will also converge to zero so the CLF constraint will be minimally satisfied. 
\end{assumption}

\begin{proposition} \label{proposition2}
Given a continous time system \eqref{eq:dynamicSystem}, there exists a constant $\tau_{CLF} > 0$ as $x_1(t) \rightarrow x_{1,d}$ and $x_2(t) \rightarrow 0$, $\forall t \geq t_k$. 
\end{proposition}
\begin{proof}
We need to show the limit of $\tau_{CLF}$ becomes a constant as the system approaches to the desired state, i.e.
\begin{equation}\label{eq:tau_limit}
    \lim_{x_1\to x_{1,d}, x_2\to 0} \tau_{CLF}
\end{equation}
Given the Assumption \ref{assumption1}, we can substitute the numerator term in \eqref{eq:tau_update_clf} and get 
\begin{equation}\label{eq:tau_clf_alternative}
    \tau_{CLF} = \frac{2\epsilon V(x(t_k))}{D} = \frac{2\epsilon V(x(t_k))}{\max_{x \in Int(C)} V''(x(t))}.
\end{equation}
In addition, there exits a closed-form solution for control $u$ with respect \eqref{eq:ECBF-ZCBF-CLF}. Given the equality assumption,  we can analytically determine the control input as
\begin{equation}\label{eq:analyticalU}
    u^* = \frac{-\epsilon V(x(t))-\pounds_{f} V(x(t_k))}{\pounds_{g} V(x(t_k))}.
\end{equation}
Since $V''(x(t))$ depends on both state $x(t)$ and control $u$. By using the closed of optimal control input \eqref{eq:analyticalU}, the numerator and denominator of \eqref{eq:tau_clf_alternative} have the same order in terms of $V(x(t))$. Therefore, $\tau_{CLF}$ becomes a constant as the system approach to desired equilibrium. 
\end{proof}
\vspace*{5mm}
The self-triggered control algorithm is summarized as
\begin{algorithm}[H]
\caption{Self-Triggered Control with CBF Constraints}
\begin{algorithmic}[1]
\Procedure {selfTriggered}{$x_0$,$K_b$,$h(x)$}
\State $x(t_k)$ := $x_0$, $\quad \forall x_0 \in Int(C)$
\While {$x(t_k) \notin Goal$}
\State Calculate optimal $u_k$ by solving \eqref{eq:ECBF-ZCBF-CLF}
\State Calculate Safe Period $\tau_{\mathrm{CBF}}$ defined in \eqref{eq:globalSafePeriod}
\State Calculate CLF Update Period $\tau_{\mathrm{CLF}}$ defined in \ref{def:CLFupdatePeriod}
\State $t_{k+1} := t_k + \min (\tau_{\mathrm{CBF}},\tau_{\mathrm{CLF}})$ 
\State For Dynamical System \eqref{eq:dynamicSystem}, hold $u_k$ between $[t_k,t_{k+1}]$
\EndWhile
\EndProcedure
\end{algorithmic}
\end{algorithm}

\section{Application to a second order integrator} \label{ExampleStudy}
In this section, we concretely apply the previous theory to the case of a simple second order integrator dynamics. Let us define the system to be
\begin{align} \label{eq:doubleIntSys}
\left[\begin{matrix}\dot{x_1}(t) \\ \dot{x_2}(t) \end{matrix} \right]= \left[ \begin{matrix} 0&1\\ 0&0 \end{matrix} \right] \left[\begin{matrix} x_1(t) \\ x_2(t) \end{matrix} \right]+ \left[ \begin{matrix} 0\\ 1 \end{matrix} \right] u.
\end{align}

Given the dynamic system \eqref{eq:doubleIntSys}, we define ECBF $h_1(x), h_2(x)$, $h_3(x), h_4(x)$, as the safety constraints:
\begin{equation}\label{eq:7}
\textbf{h}(x) = \left[\begin{matrix} h_1(x(t)) \\h_2(x(t)) \\h_3(x(t)) \\h_4(x(t))\end{matrix} \right] = \left[\begin{matrix} x_1(t)-x_{1,min} \\-x_1(t)+x_{1,max} \\x_2(t)-x_{2,min} \\-x_2(t)+x_{2,max}\end{matrix} \right],
\end{equation}
where $x_{1,min},x_{1,max},x_{2,min}, x_{2,max}$ are constants. The goal is to stabilize our system to a desired state $[x_{1,d},x_{2,d}]^T$, while still remain forward invariance, i.e. $\textbf{h}(x(t)) \geq 0,\forall t \geq t_0$. 

\subsection{CBF-CLF formulation}
Given the dynamic system $\eqref{eq:doubleIntSys}$ and set $\alpha{(h(x))} = k h(x)$, where $k$ is a relaxation constant. We have the following constraints:

\begin{equation} \label{eq:doubleIntSafetyConstraint}
\begin{aligned}
\zeta_1 &= u + k_1 x_2 + k_2 (x_1 - x_{1,min}),\\
\zeta_2 &= -u + k_1(-x_2) + k_2 (-x_1 + x_{1,max}),\\
\zeta_3 &= u + k (x_2 -x_{2,min}),\\
\zeta_4 &= -u + k (-x_2 + x_{2,max}).
\end{aligned}
\end{equation}

If $\zeta_i\geq 0, i=1,...,4$ holds, then our system is forward invariant. We define the control objective to be $x_{1,d} = 5$ and $x_{2,d} = 0$. The Lypaunov Function candidate for this particular example is
\begin{align}\label{eq:lyapunovExample}
V(x) = \left[ \begin{matrix} x_1-x_{1,d}\\ x_2 \end{matrix} \right]^T\left[\begin{matrix} 1&0.5\\0.5&1 \end{matrix}\right] \left[\begin{matrix} x_1-x_{1,d}\\ x_2 \end{matrix} \right].
\end{align}
Given \eqref{eq:CLF}, we define the CLF constraint to be
\begin{equation} \label{eq:eclf}
\eta(x) = [2x_2+(x_1-x_{1,d})]u+x_2(2(x_1-x_{1,d})+x_2)+\epsilon V.
\end{equation}
The QP formulation for system \eqref{eq:doubleIntSys} is 
\begin{equation} 
\begin{aligned}
& \underset{u\in \mathrm{U}}{\text{min}}
& & u^Tu\\
& \text{s.t.}
& & \zeta_i \geq 0, i = 1,...,4\\
& & &\eta \leq 0\\
& & &x(t_k)\in Int(C)\\
& & & u_{l} \leq u \leq u_{u}.
\end{aligned}
\label{eq:CLF-ECBF-ZCBF}
\end{equation}

\subsection{Computation of the CBF safe period}
To obtain the lower bounds for $\zeta_i$, we first calculate the derivatives $\dot{\zeta_i}$:
$\dot{\zeta_1} = k_1 x_2+k_2 u_k, \dot{\zeta_2} = -k_1 x_2-k_2 u_k, \dot{\zeta_3} = u_k$ and $\dot{\zeta_4} = -u_k$. Given the trajectory bound $r_{t_k}(t)$, we can obtain derivative bounds $\underline{\dot{\zeta_i}}$ for $\dot{\zeta_i}$. The resulting CBF constraint bounds are shown as the following:
\begin{eqnarray*}
&\underline{\zeta_1}&= (k_1(x_2(t_k)-r_{t_k}(t))-k_2\|u_k\|)t+\zeta_1(t_k), \\
&\underline{\zeta_2}&= (-k_1 (x_2(t_k)+r_{t_k}(t))-k_2 \|u_k\|)t+\zeta_2(t_k), \\
&\underline{\zeta_3}&= - k \|u_k\|t+\zeta_3(t_k),\\
&\underline{\zeta_4}&= - k \|u_k\|t+\zeta_4(t_k).
\end{eqnarray*}
\begin{remark}
Note $\underline{\zeta_i}, i=1,...,4$ do not depend on $x(t), \forall t > t_k$. We can therefore obtain safe period $\tau_i$ by directly finding the roots of $\underline{\zeta_i}$, i.e. $\underline{\zeta_i} (t_k + \tau_i)= 0$. 
\end{remark}

\subsection{Computation of CLF update period}
For an update instance $t_k$, the $\overline{V}(t)$ is obtained from Taylor-expansion at $t_k$. Given $V(x(t))$ defined in \eqref{eq:lyapunovExample},with $x_2 := x_2(t_k), x_1:=x_1(t_k)$, its first order derivative is
\begin{align*}
V'(x(t_k))&=2x_2(x_1-x_{1,d})+x_2^2+((x_1-x_{1,d})+2x_2)u_k
\end{align*}
Moreover, to find an appropriate value $D$, we obtain the second derivative as
\begin{equation*}
V''(x(t_k)) = 2x_2^2+2u_{k}(x_1-x_{1,d})+3x_2u_k+2u_k^2,
\end{equation*}
where control input $u_k$ and desired states $x_{1,d}$ are constants. 

\begin{remark} \label{VinequalityRemark}
Given the candidate Lyapunov Function \eqref{eq:lyapunovExample}, the following inequality holds
\begin{align*}
    \|x_1(t)-x_{1,d}\| \leq \sqrt{V(x(t))}, \\
    \|x_2(t)\| \leq \sqrt{V(x(t))}.
\end{align*}
\end{remark}
\vspace*{5mm}
To find the maximum value of $V''(x(t))$, we need to use Remark \ref{VinequalityRemark}. With $x_{t_k} := [x_1(t_k),x_2(t_k)]^T$, the $D$ is chosen as
\begin{align}\label{eq:Dmax}
\nonumber
D &= \max V''(x(t)) \\ &=2V(x_{t_k})+2|u_k|\sqrt{V(x_{t_k})}+3|\sqrt{V(x_{t_k})}||u_k|+2|u_k|^2
\end{align}

Next, we would like to show $\tau_{CLF}$ is finite as the system approach to the equilibrium. In practice, we do not want the controller to update infinitely fast as we approach to the desired state. After combine the result \eqref{eq:tau_clf_alternative} and substitute $u^*$ in \eqref{eq:tau_limit}, we get 
\begin{equation}
    \lim_{x_1 \to x_{1,d}, x_2\to 0} \frac{2\epsilon V(x(t))}{2V(x(t))+5\sqrt{V(x(t))}|u^*|+|u^*|^2},
\end{equation}
where the numerator and denominator have the same rate of converging as $x_1 \to x_{1,d}$ and $x_2 \to 0$. Therefore, $\lim \tau_{CLF}$ is a constant when the system approaches to equilibrium.

Now we can obtain $\overline{V}(t)$ that is defined in \eqref{eq:vBar}, where $D$ is calculated using \eqref{eq:Dmax} at each update instance. The CLF update period is 
\begin{equation}
    \tau_{CLF} = \frac{-2(2x_2(x_1-x_{1,d})+x_2^2+((x_1-x_{1,d})+2x_2)u_k)}{D}.
\end{equation}

\subsection{Simulation}
Given the double integrator system \eqref{eq:doubleIntSys} and an initial state $x_0$, the objective is to reach $x_{1,d} = -7$, $x_{2,d} = 0$. The two approaches: self-triggered and periodic controls, are both used for comparison. We define self-triggered control updating interval as $t_s$ and periodic control updating interval as $t_p$. At each controller update instance $t_k$, the CBF-CLF Quadratic Program \eqref{eq:CLF-ECBF-ZCBF} is solved using quadprog() function in Matlab 2018a with Core i5-8259U CPU. The elapsed time for solving each QP problem is around 0.0019s, which is much smaller than the required update interval. In Table~\ref{table}, the experiment parameters are defined. 

\begin{table}[ht]
\centering
\caption{CLF-CBF Controller Parameters}\label{table}
\begin{center}
\scalebox{0.9}{
\begin{tabular}{|c|c|}
\hline
\textbf{Parameter}&\textbf{Value} \\
\hline
$x_0$&$[6, 5]^T$ \\
\hline
$x_{1,min}$&$-10$ \\
\hline
$x_{1,max}$&$10$ \\
\hline
$x_{2,min}$&$-10$ \\
\hline
$x_{2,max}$&$10$ \\
\hline
$\epsilon$&0.8 \\
\hline
$L$& 1\\
\hline
$K_b$&[105 20.5] \\
\hline
$[u_l, u_u]$&[-20, 20] \\
\hline
$t_p$ & 0.75 \\
\hline
$t_s$ & $\min(\tau_{\mathrm{CBF}},\tau_{\mathrm{CLF}})$ \\
\hline
\end{tabular}}
\label{tab1}
\end{center}
\end{table}

The result is illustrated in figure \ref{fig:selfTriggered} and \ref{fig:periodic}. For self-triggered control, it is clear that the controller only updates when the system is about to violate CBF constraints or the system is deviating away from the desired states. Notice that, the update interval for self-triggered controller becomes a lot faster as the system approaches to the unsafe region ($x_1 < x_{1,\min}$) in order to prevent violation on safety constraint. Moreover, the CLF update period converges to 0.3166s and remains as a constant as system approaches to equilibrium, which validate the proposition \eqref{proposition2}. In the periodic controller case, the position $x_1$ violates $x_{1,min}$ constraint for $t \in [3,4]$. (See Figure \ref{fig:trajComparison} for a different perspective). It clearly demonstrates the issues with the periodic controller in real-life situation i.e., the controller neither knows the correct sampling rate in-advance, nor has the ability to adjust it real-time. All the safety and convergence properties from CBF and CLF formulation could fail when applying the controller in this manner.

\begin{figure}[ht]
    \centering
    \subfigure[]{
    \includegraphics[width=0.3\textwidth]{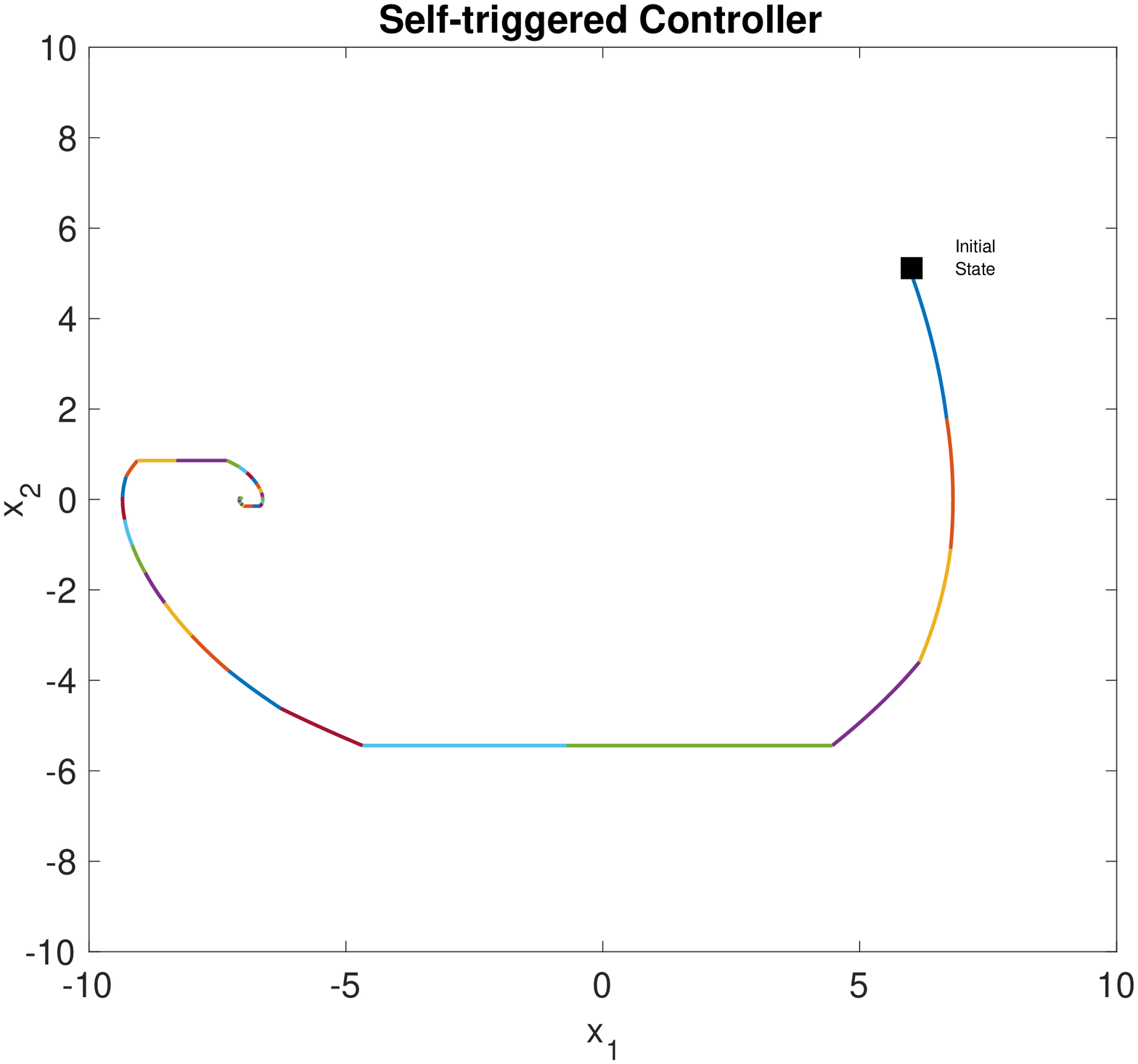}}
    \subfigure[]{
    \includegraphics[width=0.3\textwidth]{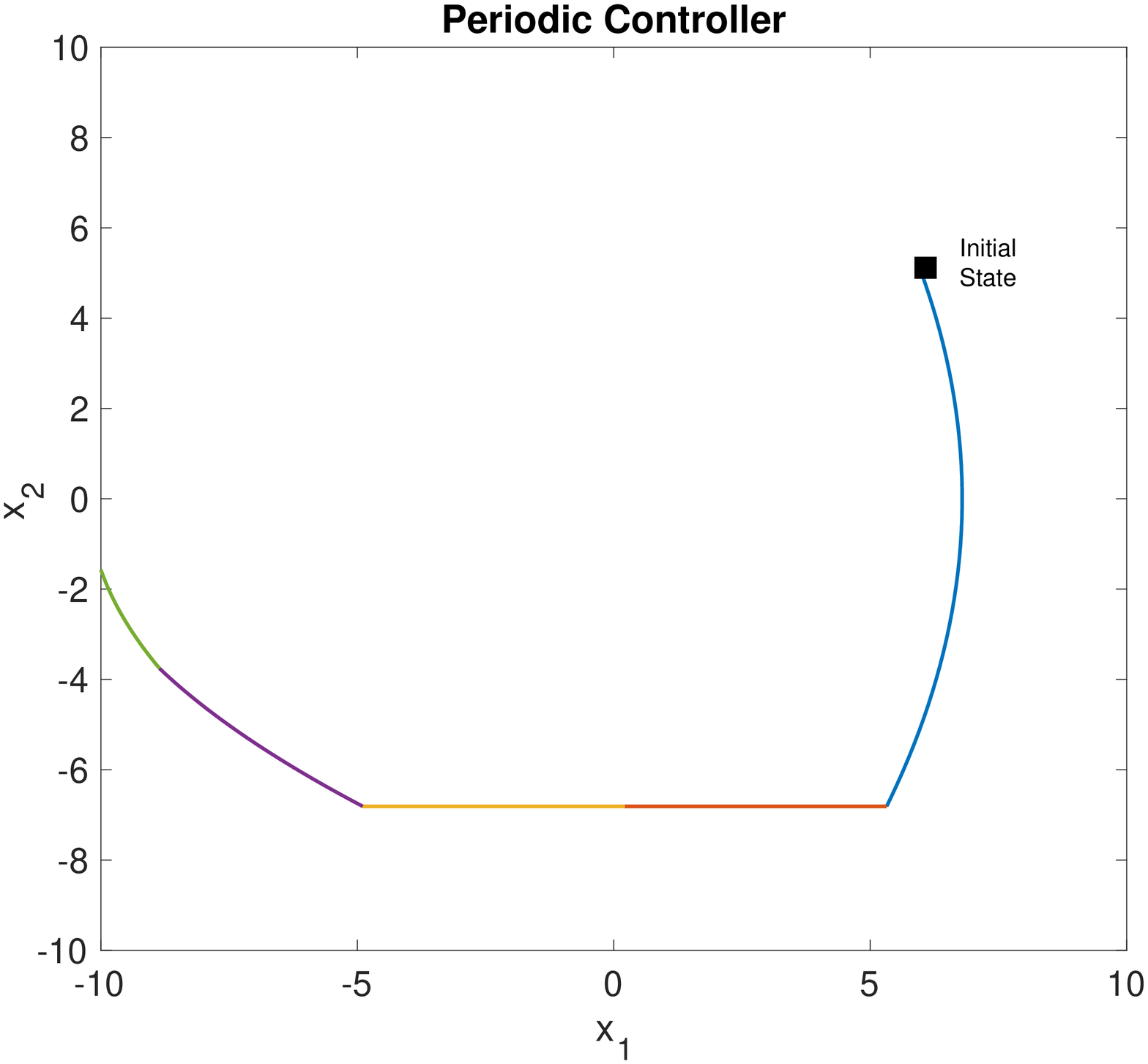}}
    \caption{System Trajectory Comparison}
    \label{fig:trajComparison}
\end{figure}

\subsection{Discussion}\label{Discussion}
The self-triggered controller is designed to obtain the safety property from CBFs while stabilizing a system asymptotically. In addition, we use the ECBF framework to ensure the controller works for system with high relative degrees. In this example, since we cannot directly control the state $x_1$, the use of ECBF becomes necessary. The nice relation that is shown in Remark \ref{remark1} between ZCBF and ECBF also reduces the complexity when designing this controller. The proposition \eqref{proposition2} guarantees that the controller can be applied in a real-world situation where the the update frequency cannot be infinitely fast and it is numerically validated in the experiment. Although the example study focuses on a simple double integrator system, we believe that the idea can be applied to nonlinear system in general. That being said, the calculation of lower bounds for CBFs and upper bounds for Lyapunov function might not be trivial. Moreover, like all model based controllers, the CLF-CBF controller relies on an accurate system model to work well. 

\section{Conclusion} \label{Conclusion}
In this paper, we proposed a self-triggered controller with CBF-CLF based QP formulation guarantee the safety of our system under ZOH mechanism. It is a starting point to bridge the gap between theoretical work and real-life implementation with the involvement of digital computers. This novel approach has been successfully validated on a double integrator dynamics. In addition, the theoretical contribution includes trajectory bound with system dynamics and proof of fixed update interval as the system approaches to equilibrium. For future work, we would like to look into problems with more complicated dynamics, such as quad-copters and manipulators. In addition, the effect of external disturbances will also be studied.

\begin{figure*}
\center
\includegraphics[width=0.97\textwidth]{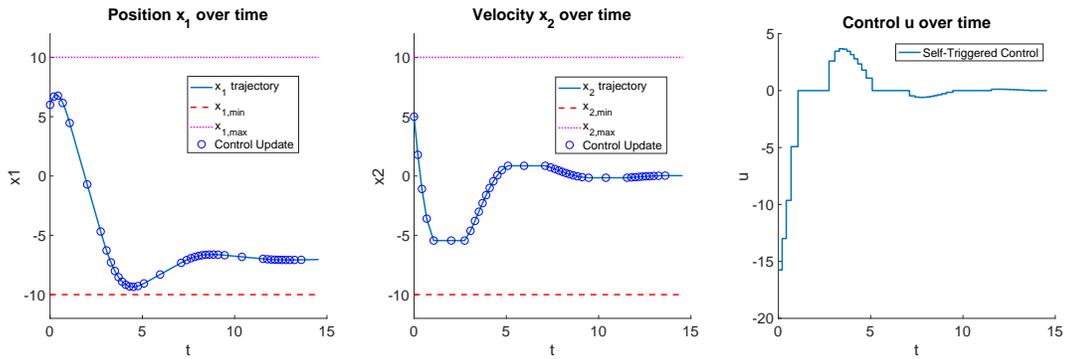} 
\caption{\small \sl Self-triggered control with variable time step}
\label{fig:selfTriggered}
\end{figure*}
\begin{figure*}
\center
\includegraphics[width=1\textwidth]{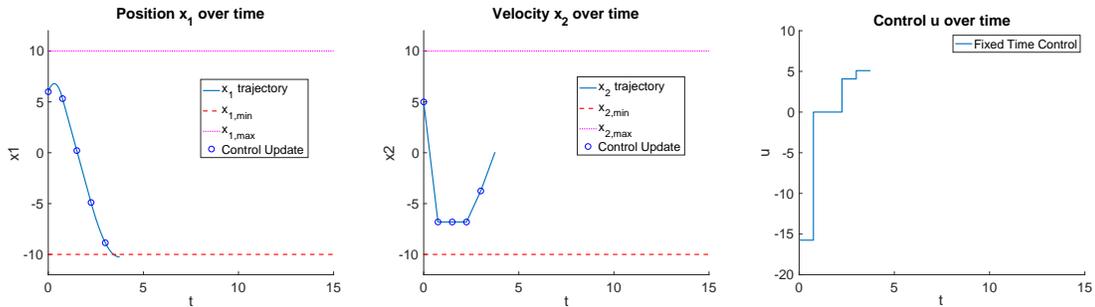}
\caption{\small \sl Periodic control with constant time step}
\label{fig:periodic}
\end{figure*}

\bibliographystyle{IEEEtran}
\bibliography{IEEEabrv,mybib}

\end{document}